\documentclass{article}

\usepackage[final, nonatbib]{nips_2016}

\usepackage[utf8]{inputenc} 
\usepackage[T1]{fontenc}    
\usepackage{hyperref}       
\usepackage{url}            
\usepackage{booktabs}       
\usepackage{amsfonts}       
\usepackage{amsthm}         
\usepackage{amsmath}
\usepackage{nicefrac}       
\usepackage{microtype}      
\usepackage{geometry}
\usepackage{footnote}
\makesavenoteenv{tabular}
\makesavenoteenv{table}
\usepackage[pdftex]{graphicx}
\usepackage[numbers]{natbib}

\usepackage[shortlabels]{enumitem}

\newcommand*\Eval[3]{\left.#1\right\rvert_{#2}^{#3}}
\DeclareMathOperator*{\argmax}{arg\,max}

\newtheorem{prop}{Proposition}[section]
\newtheorem{thm}{Theorem}[section]

\newtheorem{lemma}[thm]{Lemma}

\newcommand{\angstrom}{\mbox{\normalfont\AA}}

\title{Conformation Clustering of Long MD Protein Dynamics with an Adversarial
       Autoencoder}

\author{
  Yunlong Liu\\
  Department of Biophysics and Biophysical Chemistry\\
  The Johns Hopkins University\\
  Baltimore, MD 21205 \\
  \texttt{yliu120@jhmi.edu} \\
  \And
  L. Mario Amzel \\
  Department of Biophysics and Biophysical Chemistry\\
  The Johns Hopkins University \\
  Baltimore, MD 21205 \\
  \texttt{mamzel@jhmi.edu} \\
}

\begin{document}
\bibliographystyle{plain}

\maketitle

\begin{abstract}

  Recent developments in specialized computer hardware have greatly accelerated
  atomic level Molecular Dynamics (MD) simulations. A single GPU-attached
  cluster is capable of producing microsecond-length trajectories in reasonable
  amounts of time. Multiple protein states and a large number of microstates
  associated with folding and with the function of the protein can be observed
  as conformations sampled in the trajectories. Clustering those conformations,
  however, is needed for identifying protein states, evaluating transition rates
  and understanding protein behavior. In this paper, we propose a novel
  data-driven generative conformation clustering method based on the adversarial
  autoencoder (AAE) and provide the associated software implementation
  \textbf{Cong}. The method was tested using a 208 $\mu$s MD simulation of the
  fast-folding peptide Trp-Cage (20 residues) obtained from the D.E. Shaw
  Research Group. The proposed clustering algorithm identifies many of the
  salient features of the folding process by grouping a large number of
  conformations that share common features not easily identifiable in the
  trajectory.

\end{abstract}

\section{Introduction}

With the development over the past decade of high performance CPU clusters, GPU
accelerated computing and software-level optimization, the performance of
general molecular dynamics (MD) software (NAMD\cite{NAMD},
GROMACS\cite{Gromacs}, AMBER\cite{amber}, CHARMM\cite{charmm}, etc.) has been
extraordinarily enhanced. Increasingly, individual research groups are capable
of running microsecond-scale MD simulations for macromolecular systems of
considerable size. Some research laboratories equipped with specialized hardware
or large-scale distributed software \cite{ANTON}\cite{FoldingAtHome} can even
produce millisecond-scale trajectories. Considering an MD trajectory as a number
of conformations sampled from the complete ensemble (NVE, NVT or NPT) as it is
observed through its time evolution, long trajectories provide significant
explorations of the conformational space and can reveal key conformational
transitions of the macromolecules. Depending on the simulation, those
conformational states may reflect stages in the folding of the protein or
details about the mechanism of action of the molecules. Objective unbiased
identification of these states is crucial for extracting mechanistic information
from MD simulations.

Though in some long MD trajectories conformational changes may sometimes be
visually recognized, identifying key conformations, grouping similar
conformations and quantifying key conformational transitions are more
objectively dealt with using specialized software. This is because the number of
frames in long trajectories is very large and the dimensionality of the
simulation system is high. However, grouping conformations and identifying
transitions are often essential for
understanding the behavior of the molecules and for comparing changes of their
behavior under different environments or those that result from mutations. To
tackle these tasks, similarity-based clustering algorithms, such as KMeans,
KMedoids, hiearchical clustering, agglomerative clustering and other
domain-specific algorithms are used. Their implementations are available as
utilities in most major MD packages\cite{Gromacs}\cite{MDTraj}\cite{MDAnalysis}.
All these similarity-based algorithms share the following common basis:
clustering procedures depend on frame-wise similarities, usually quantified by
the RMSD. Except for KMeans and KMedoids, which are optimized by expectation
maximization, all the other algorithms additionally rely on a pre-determined
cutoff to make decisions about how to assign individual frames to a cluster.

Similarity-based clustering algorithms solve the clustering task successfully
but a few caveats remain. As MD trajectories are getting longer, the number of
frames, denoted by $N$, is growing by orders of magnitude. Since the computation
and space complexity of the RMSD matrix are both $O(N^2)$, algorithms based on
calculating the RMSD matrix take large amounts of computing time and can
overwhelm the memory. In practice, many researchers run clustering using every
$n$-th frame of the long trajectory; however, this strategy could introduce
significant bias into calculations such as those of transition probabilities
between clusters. In addition, since RMSDs are essentially Euclidean distances
defined on a high dimensional space, clustering using this metric could be
vulnerable to artifacts due to their sensitivity to large variations in the
conformations of the termini and of flexible loops\cite{Metrics}\cite{Survey}.
Some domain-specific RMSD-based clustering methods may suffer from additional
side-effects.  For instance, the GROMOS algorithm\cite{Gromos} always
concatenates two segments of the time series after extracting a neighborhood as
a new cluster, which results in a cutoff-sensitive segmentation of the original
trajectory. Given these caveats, it is clear that an ideal clustering method
should have the following properties: 1) the time and space complexity should be
proportional to $N$; 2) since defining an ideal distance metric for measuring
similarity between structures is very difficult, the algorithm should be data-driven
rather than relying on a pre-defined metric. In other words, ideally, the
algorithm would "learn" the metric by optimizing some loss function. 3) The
algorithm should always look at the entire dataset to avoid unnecessary
segmentations.

Recently, deep neural network (dNN) models have been proven to be very
successful in the field of Artificial Intelligence (AI). dNN models are powerful
in developing very complicated non-linear functions that can map a set of raw
image pixels in the data space to some separable regions in the feature space
and work conversely as generative models when the probabilistic distributions of
the latent representations are available. The training process is entirely
data-driven and converges after feeding the entire dataset for a limited
number of times ("epochs"). Although dNN models are primarily employed in AI,
they are intrinsically general and can be applied to problems in other fields.
In this paper, we exploit the data-driven property of dNN, to propose a novel
method for clustering MD trajectories based on an existing dNN model --
Adversarial Autoencoder (AAE)\cite{AAE}.

\section{Theory}

\subsection{Deep Autoencoder}

Deep autoencoders (DAE) are widely used for unsupervised learning tasks such as
learning deep representations or dimensionality reduction. Typically, a
traditional deep autoencoder consists of two components, the encoder and the
decoder. Each component is constructed with a feed-forward neural network,
which,
in essence, represents a non-linear function. Let's denote the encoder's
function as $f_{\theta}: \mathcal{X} \rightarrow \mathcal{H}$, and denote the
decoder's function as $g_{\omega}: \mathcal{H} \rightarrow \mathcal{X}$, where
$\theta, \omega$ are parameter sets for each function, $\mathcal{X}$ represents
the data space and $\mathcal{H}$ the feature (latent) space. The autoencoder
used in the model presented here is always aiming at mapping the high-dimensional
data space to a lower-dimension feature space and the quality of the
mapping is monitored by how well it reconstructs the original data space from
the feature space. For this purpose, it optimizes the reconstruction loss, which
is

\begin{equation}
L(\theta, \omega) = \frac{1}{N}\left\Vert X - g_{\omega}(f_{\theta}(X)) \right\Vert^2
\end{equation} where $L(\theta, \omega)$ represents the loss function for the
reconstruction.

The use of the $L^2$-norm in the reconstruction loss is valid for guiding the
reconstructions to remain close to the original data. This approach may be
contrasted with the $L^2$-norm used to measure similarities between frames. In
fact, for a finite dimensional vector space $V$, all $l_p$ norms ($p \geq 1$)
are equivalent and induce the same topology. Assuming the dimensionality
reduction with DAE is successful, the deep representations generated will
feature the major skeleton of MD frames. Therefore, discriminating frames by
their deep representations is inherently insensitive to the random noise.
Moreover, if we assume a given distribution for the deep representations, the
reconstruction function $g_{\omega}$ naturally serves as a generating function
that converts the model into a generative model.

\subsection{Generative Adversarial Network}

Generative Adversarial Network (GAN)\cite{GAN} is designed to be a deep generative model.
It assumes that all observed data are generated from some latent random variable
$\mathbf{z}$ with a known prior distribution
$p(\mathbf{z})$. The overall goal of GAN is to learn how to generate
$\mathbf{x}$ from $\mathbf{z}$, e.g. finding
$p(\mathbf{x}|\mathbf{z})$. A GAN model has two components: a generator network
(denoted by function $G(\mathbf{z}; \theta_g)$) and a discriminator network
(denoted by function $D(\mathbf{x}; \theta_d)$), which are non-linear functions
constructed with multilayer networks. Unlike the autoencoder and many
traditional models, GAN is a two-player game rather than an optimizer of a
monolithic loss function. The discriminator tries its best to differentiate the
real data samples from the generated fake samples, while the generator tries to
fool the discriminator by generating counterfeits. When the discriminator is
unable to distinguish the source of the input samples, the generator is considered
trained. In order to achieve this goal, GAN trains the discriminator and the
generator simultaneously in two phases:

\begin{equation}
\label{eq:gan}
\min_G\max_D L(D, G)=\mathbb{E}_{\mathbf{x}\sim
\mathbb{P}_{data}(\mathbf{x})}[f(D(\mathbf{x}))]   + \mathbb{E}_{\mathbf{z}\sim
\mathbb{P}_{\mathbf{z}}(\mathbf{z})}[f(1-D(G(\mathbf{z})))]
\end{equation} where $f: \mathbb{R} \rightarrow \mathbb{R}$ and $\mathbb{E}[\cdot]$
denotes the expectation of a random variable.

Standard GAN chooses $f$ as the logarithmic function and the above adversarial
loss function adopts the form of the \textit{Jensen-Shannon} (JS) divergence.
Training GAN with the JS divergence loss is very likely to fall into unstable
traps such as mode collapse. A recently proposed variant of GAN, Wasserstein GAN
(WGAN)\cite{WGAN}, replaces the JS divergence loss with Wasserstein loss, which is,

\begin{equation}
\label{eq:wgan}
  W(\mathbb{P}_{data}, \mathbb{P}_{G(\mathbf{z})}) = \sup_{||D||_L \leq K}
  \mathbb{E}_{\mathbf{x} \sim \mathbb{P}_{data}}[D(\mathbf{x})] - \mathbb{E}_{\mathbf{x}\sim
  \mathbb{P}_{G(\mathbf{z})}}[D(\mathbf{x})]
\end{equation}

Essentially, WGAN substitutes $f$ in Equation \ref{eq:gan} with an identity
function and this removes the sigmoid activation function in the last layer of
the discriminator network. Empirically, WGAN proved to be more stable but it has
to be trained with a fairly small step and weight clipping to keep the
discriminator function $D(x)$ approximating a $K$-Lipschitz function. A recent
study\cite{WGAN_GP} shows that weight clipping leads to optimization
difficulties and could result in pathological behaviors. Instead of tuning an
optimal clipping threshold, the study \cite{WGAN_GP} suggests enforcing the
Lipschitz constraint by adding a gradient penalty term to the original loss
function given by Equation \ref{eq:wgan}. In our experiments, we train both
WGAN components in our model with the gradient penalty (Term $\lambda
\mathbb{E}_{\hat{\mathbf{x}} \sim \mathbb{P}_{\hat{\mathbf{x}}}}
[(\left\Vert\nabla_{\hat{\mathbf{x}}} D(\hat{\mathbf{x}}) \right\Vert_2 - 1)^2]$
in Equation \ref{eq:wgan-gp}), rather than using weight clipping:

\begin{equation}
\label{eq:wgan-gp}
L(\mathbf{x}) = \mathbb{E}_{\mathbf{z} \sim \mathbb{P}(z)}[D(G(\mathbf{z}))]
- \mathbb{E}_{\mathbf{x} \sim \mathbb{P}_{data}}[D(\mathbf{x})] +
\lambda \mathbb{E}_{\hat{\mathbf{x}} \sim \mathbb{P}_{\hat{\mathbf{x}}}}
[(\left\Vert\nabla_{\hat{\mathbf{x}}} D(\hat{\mathbf{x}}) \right\Vert_2 - 1)^2]
\end{equation} where $\hat{\mathbf{x}} = \epsilon \mathbf{x} + (1-\epsilon) G(\mathbf{z})$,
and $\epsilon \sim U[0, 1]$.

\subsection{Adversarial Autoencoder}

Generally, AAEs have different architectures targeted at different
tasks. Since our goal is to use it as a generative clustering method, we adopt
the unsupervised version of AAE \cite{AAE}. This version of the architecture
(see Figure 1) is built based on the data generative process it assumes: each
data element $\mathbf{x}$ in the dataset is generated from a latent categorical
random variable $\mathbf{y}\sim Cat(\boldsymbol{\pi})$, represented as a
$K$-dimensional one-hot vector ($K$ is the dimension of $\mathbf{y}$ as well as
a user-defined cluster number) and a Gaussian style variable $\mathbf{z} \sim
N(0, I)$. The autoencoder component of the AAE estimates the encoding function
$p(\mathbf{y}, \mathbf{z}| \mathbf{x})$ and the decoding function
$p(\mathbf{x}|\mathbf{y}, \mathbf{z})$ by minimizing the reconstruction loss.
Here we denote the encoding function and decoding function as $p(\cdot)$, since
in general, $p(\cdot)$ could be modeled as either deterministic functions or
probabilistic densities. However, in this paper \cite{AAE}, we model these
functions as deterministic. Following the original paper, we denote the estimate
of $p(\cdot)$ as $q(\cdot)$. While optimizing the reconstruction loss, the
category GAN and the style GAN regulate the autoencoder's encoder by imposing
the prior distribution of $\mathbf{y}$ and $\mathbf{z}$ onto $q(\mathbf{y})$ and
$q(\mathbf{z})$, where $q(\cdot)$ is defined as,

\begin{equation}
q(\cdot) = \int_{X}q(\cdot|x)\mathbb{P}(\mathbf{x})d\mathbf{x}
\end{equation}

With these considerations, we can write the reconstruction loss as,

\begin{equation}
L = \underset{\mathbf{x} \sim \mathbb{P}(\mathbf{x})}{\mathbb{E}}
\underset{\substack{\mathbf{y} \sim q_{cat}(\mathbf{y}) \\
\mathbf{z} \sim q_{style}(\mathbf{z})}}{\mathbb{E}}
\left\Vert g_{s}(g_{cat}(\mathbf{y}) + g_{style}(\mathbf{z})) - \mathbf{x} \right\Vert^2
\end{equation}
where $g_s$, $g_{cat}$ and $g_{style}$ denote the decoding function of the
shared decoder, the categorical decoder and the style decoder, respectively.

Intuitively, the autoencoder guards the quality of the latent features with its
reconstruction loss. The category GAN pushes the categorical feature vector to
$K$-simplex and the style GAN leads $p(\mathbf{z})$ to a standard Gaussian. In
other words, AAE assigns an expected $N\pi_i$ number of frames into the $i$th bucket
and each bucket has a standard Gaussian shape. Based on the fact that AAE can
impose prior distributions to latent representations successfully, it
can feed samples from the prior distributions to the trained generative decoder
to create new conformations that do not exist in the original trajectories.

A serious drawback of this method is that it needs to have a correct expectation
of the categorical mass $\mathbf{\pi}$, which is usually impossible to estimate
in practice. Intuitively, imposing an incorrect prior knowledge onto the
categorical encodings may harm the reconstruction loss, and therefore the
quality of the clustering. It remains hard to mathematically estimate the
penalty caused by the incorrect prior due to the complex nature of the
optimization of this model. A "work-around" solution to this issue is to define
a large number of small clusters and assign a uniform distribution to them. This
solution can remove the potential penalty to some extent. However, additional
analytical steps are needed for further grouping the resulting small clusters
and, specifically in MD cases, for understanding the estimated transition
probability matrix.

Our proposal is the reparameterization of the categorical prior with a Gumbel-Softmax
distribution \cite{Concrete}\cite{CVAE} to eliminate the need to specify
the categorical mass beforehand. In other words, we propose a special AAE model
for clustering that does not need any prior knowledge of how examples are
distributed among clusters.

\subsection{Adversarial Autoencoder with Gumbel-Softmax}

\subsubsection{Reparameterization Trick}

\textit{Reparameterization Trick} refers to a sampling process in the machine
learning and statistics world. The major aim of using the
reparametrization trick is to separate the stochastic sampling part and
parameter-dependent transformation part of a parameterized distribution. In
stricter terms, it assumes that there is a parameterized distribution
$D(\alpha)$, where $\alpha$ is the parameter. To reparameterize $D(\alpha)$, we
first sample from a unparameterized distribution $Z$ and then apply a
deterministic function that depends on the parameter $f_{\alpha}(\mathbf{z})$ to
the samples to make them as directly  sampled from $D(\alpha)$.

After uncoupling the parameters from the stochastic sampling, the distribution
parameters become trainable model parameters since their gradients can be
estimated by calculating the partial derivative of the loss function with
respect to them. In our case, we rewrite the reconstruction loss by
reparameterizing the categorical prior,

\begin{equation}
L = \underset{\mathbf{x} \sim \mathbb{P}(\mathbf{x})}{\mathbb{E}}
\underset{\substack{\boldsymbol{\alpha} \sim q_{cat}(\boldsymbol{\alpha}) \\
\mathbf{z} \sim q_{style}(\mathbf{z})}}{\mathbb{E}}
f(h(\boldsymbol{\alpha}, \boldsymbol{\pi}), \mathbf{z})
\end{equation}
where $f(\mathbf{y}, \mathbf{z}) = \left\Vert g_{s}(g_{cat}(\mathbf{y}) +
g_{style}(\mathbf{z})) - \mathbf{x} \right\Vert^2$ and we use $h(\boldsymbol{\alpha},
\boldsymbol{\pi})$ to substitute $\mathbf{y}$. In addition, we use
$\boldsymbol{\alpha}$ to represent a random vector with a fixed distribution
and $\boldsymbol{\pi}$ to represent the categorical mass. The function $h(\cdot)$
is the transformation function.

Based on Equation 7, we can obtain the gradient with respect to $\boldsymbol{\pi}$,

\begin{equation}
\nabla_{\boldsymbol{\pi}} L =
\underset{\mathbf{x} \sim \mathbb{P}(\mathbf{x})}{\mathbb{E}}
\underset{\substack{\boldsymbol{\alpha} \sim q_{cat}(\boldsymbol{\alpha}) \\
\mathbf{z} \sim q_{style}(\mathbf{z})}}{\mathbb{E}}
\nabla_{h} f(h(\boldsymbol{\alpha}, \boldsymbol{\pi}), \mathbf{z})
\nabla_{\boldsymbol{\pi}} h(\boldsymbol{\alpha}, \boldsymbol{\pi})
\end{equation}

The \textit{Gumbel-Max Trick}\cite{Gumbel}\cite{ASampling} is a reparameterization
of a one-hot categorical distribution. It samples a $d$-dimensional vector from
a standard Gumbel distribution $\boldsymbol{\alpha}$ and uses $f(\boldsymbol{\alpha})
= \operatorname{onehot}(\operatorname{argmax}(\boldsymbol{\alpha} + \log{
\boldsymbol{\pi}}))$ as the transformation function. Since the $\operatorname{argmax}$
function is not a differentiable function, the gradient cannot propagate to
the parameter $\boldsymbol{\pi}$. The Gumbel-Softmax distribution aims to relax
the Gumbel-Max Trick by replacing the $\operatorname{argmax}$ function with
$\operatorname{softmax}$.

\begin{equation}\label{sample}
y_i = \frac{\exp{((\log{\pi_i} + \alpha_i) / \tau)}}
{\sum_{j=1}^{k}\exp{((\log{\pi_j} + \alpha_j) / \tau)}}
\end{equation}
where $\tau \in (0, +\infty)$ is the temperature parameter.

This distribution was originally introduced as Concrete Distribution by Maddison
et al.\cite{Concrete} and was successfully used in a related work\cite{CVAE}. The
distribution also relieves the constraint that the categorical mass should belong to
the $(k-1)$-dimensional simplex. $\pi_i$s are not necessarily normalized, so they
can be anywhere on $\mathbb{R}_{+}$. When directly learning the logits of $\pi_i$s,
no constraint is needed since the logits belong to $\mathbb{R}$. The extra parameter
$\tau$ grants the distribution a few interesting properties. These properties are
proved in Supplementary Info A.

\begin{prop}\label{prop}
Let $\boldsymbol{y}$ be a random variable sampled from Equation \ref{sample}, where
$\boldsymbol{\pi} \in (0, +\infty)^d$, $\tau \in (0, +\infty)$, then we have, \\
\begin{enumerate}[(1)]
  \item $\forall k$, $\lim\limits_{\tau \rightarrow +\infty} y_k = \frac{1}{d}$. In fact,
  $\lim\limits_{\tau \rightarrow +\infty} y_k$ defines a random variable, which degenerates to a
  constant $\frac{1}{d}$ everywhere in the event space $\Omega$, which is $\mathbb{R}^k$.
  \item $\lim\limits_{\tau \rightarrow 0^{+}} \mathbf{y} \sim \text{OneHotCategorical}
  (\frac{\boldsymbol{\pi}}{\sum_{i=1}^{k}\pi_i})$
  \item $\tau \rightarrow 0^{+}$, $\mathbb{E}{y_i} = \frac{\pi_i}{\sum_{j=1}^{k}\pi_j}$.
\end{enumerate}
\end{prop}

\subsubsection{AAE with Gumbel-Softmax}

A general architecture of our model, namely, AAE with Gumbel-Softmax (AAE-GS), is
illustrated in Figure \ref{fig:aae-gs}. By reparameterizing the categorical prior
with the Gumbel-Softmax distribution, we extract the categorical mass as a
trainable variable, which enables us to impose only the "stochastic" part of the
categorical distribution. In this case, the logits of the categorical mass will
be updated each time the autoencoder's loss is updated.

To train the model, instead of assuming a categorical mass, we define an
annealing process of the parameter $\tau$ in Equation (\ref{sample}). First,
we set $\tau$ to a large value to simulate $\tau \rightarrow
+\infty$. By proposition \ref{prop}, we know the model will approximately
put every example in only one cluster, since the categorical representation for
each example is almost same. We then anneal $\tau$ to a value close to zero in
some finite number of steps so that eventually the model can assign each example a one-hot
vector marking the most suitable bucket to put it in.

\begin{figure}[h]
  \centering
  \includegraphics[width=5in]{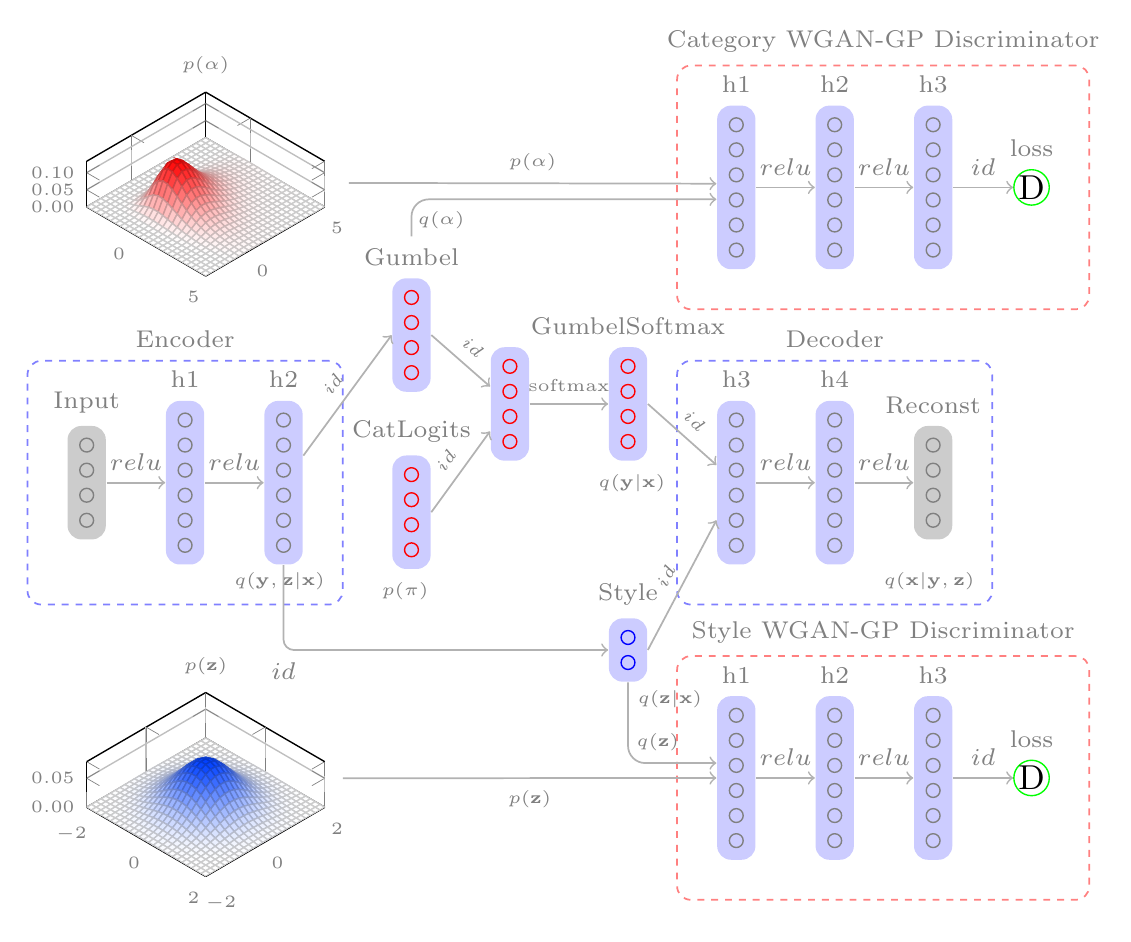}
  \caption{A general architecture of AAE with Gumbel-softmax. The architecture consists
  of four major components: the encoder, the decoder, the discriminator for categorical
  GAN and the discriminator for style GAN. Each component is represented as a neural network,
  which is represented as a multilayer perceptron in this figure. A general choice of
  activation functions applied to all layers are marked above the arrows connecting all
  layers. "relu" is short for Rectifier Linear Unit (ReLU) and "id" is short for the identity
  function. "h1", "h2", etc represent hidden layers in a neural network component. The
  middle row of the architecture essentially forms a deep autoencoder whose deep representation
  is divided into a categorical part (red circles) and a style part (blue circles). The
  encoder combined with the top row forms the categorical GAN and with the bottom row
  forms the style GAN. $p(\alpha)$ illustrates a Gumbel(0, I) distribution and $p(\mathbf{z})$
  illustrates a standard Gaussian.}
  \label{fig:aae-gs}
\end{figure}

\section{Results}

\subsection{Trp-Cage}

Unfortunately, a standard dataset for examining the clustering performance of
different algorithms for MD frames has not been established yet. To test our
algorithms, a model trajectory with the following properties was required: 1) It
samples \textit{a small protein}, 2) It is a long enough simulation that
traverses a large scope of the protein's conformation space, 3) It samples
conformations of the protein that are easily distinguishable by the human eye.

A few years ago, Lindorff-Larsen et al. \cite{fastfolding} did a thorough MD
based folding study on 12 small proteins with pure MD simulations. It was
observed that each protein folded to a conformation very close to its
experimentally determined structure during those simulations, and that the
process of folding occurred many times during the simulation. There were 9
fast-folding proteins amongst the 12 selected by Lindorff-Larsen et al. We
picked a small protein, \textit{Trp-cage} (PDB ID: 2JOF) from these 9
fast-folders.  Trp-cage is a 20-residue protein that folds as two small
$\alpha$-helices and a short loop. The fold is stabilized by having a tryptophan
residue (Trp6) at the core of the protein with its indole side chain sandwiched
between two prolines (Pro12 and Pro18) and surrounded on the other sides by a
tyrosine (Tyr3, a leucine (leu7) and another proline (Pro19). This protein has a
folding time of 14 $\mu$s and was simulated for a total of 208 $\mu$s.  Compared
to other fast-folding proteins provided by the study, the folded structure of
Trp-cage has clear secondary structure elements that can fold into different
conformations even though its length is relatively small.  Additionally, its
short folding time makes the size of the trajectory compact.  For our purposes,
this protein is an ideal model for testing the clustering algorithms. The
trajectory was generously provided by \textit{D.E. Shaw Research}.

The trajectory we obtained contains a total of 1,044,000 frames assembled in 105
shards. In our clustering experiments, we only considered the backbone coordinates
of the protein. Therefore, we sliced each frame by the backbone selection. The
selection includes 80 atoms of Trp-Cage so that flattening all the coordinates
per frame of these atoms results in a 240-dimensional array. The full trajectory
was aligned to a specific frame and translated into TFRecords format
specifically designed for Tensorflow\cite{Tensorflow} Framework and resharded
into 10 shards. We used this dataset for all our experiments.

\subsection{Clustering with AAE}

\begin{figure}[h]
  \centering
  \includegraphics[width=5.5in]{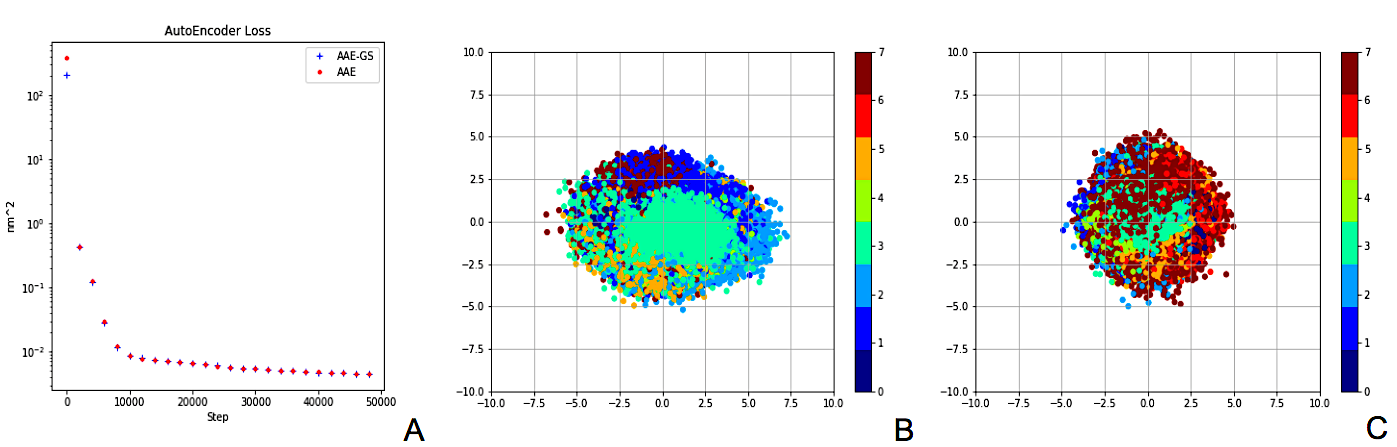}
  \caption{Convergence of AAE and AAE-GS clustering experiments.
  The restruction losses plotted in A of both experiments were evaluated every 2000 steps.
  B and C plots the first two dimensions of the "learned" style vectors of all frames
  in AAE and AAE-GS experiments, respectively.}
  \label{convergence}
\end{figure}

We performed a clustering experiment on our dataset using the general
unsupervised architecture of AAE, assuming the categorical prior obeys the
discrete uniform distribution and using 8 as the number of clusters. We set the
dimension of the style representation to 16 and all its components as
independent identically distributed (i.i.d) standard Gaussians. The
hyperparameters and the architecture used for this experiment are listed in
Table 1 of Supplementary Info B. We trained the model until the autoencoder's loss
converged, which took $>24$ hours on one NVIDIA Tesla K80 GPU card. (See Figure
\ref{convergence} A). The final reconstruction loss was 0.005 $nm^2$, so the
average difference between the coordinate of any atom and that in its
reconstruction of that atom is approximately 0.7 $\angstrom$.

\begin{figure}[h]
  \centering
  \includegraphics[width=5in]{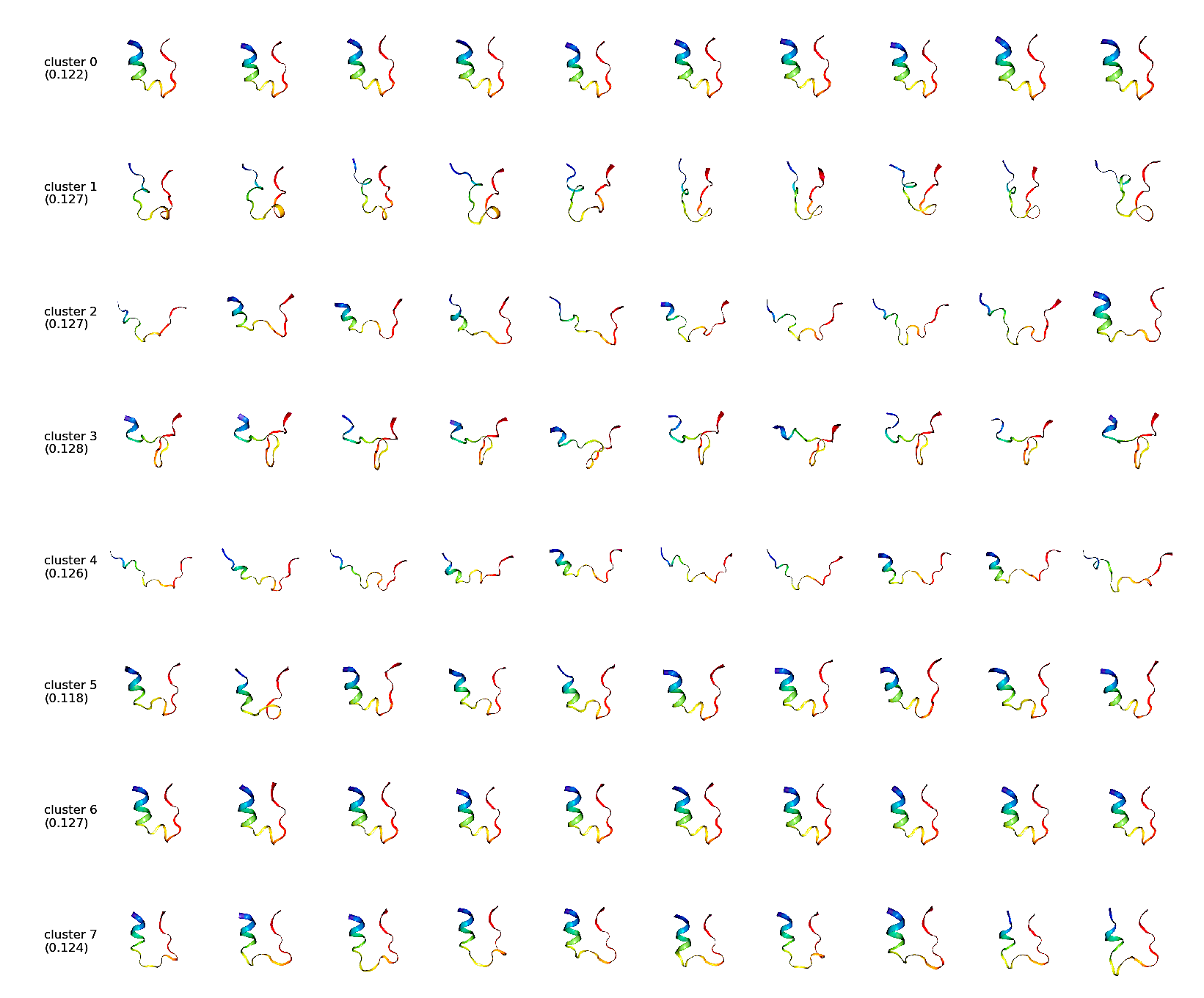}
  \caption{Visualization of cluster representatives obtained from AAE experiment.
   The first ten frames of each cluster sorted by the norm of the style vector
   are shown in each row. All structures are oriented to a fixed orientation.
   The estimated probability mass for each cluster is listed
   in brackets underneath each cluster id label.}
  \label{aae_clustering_result}
\end{figure}

To better visualize frames in each cluster, we sorted the frames in each cluster
by the norm of their style vectors from small to large,
since we consider that the smaller the norm of the noise the more representative a
frame is (Figure \ref{aae_clustering_result}).

The results show that the uniform categorical prior is successfully imposed on
the categorical representation $\mathbf{y}$. The first two dimensions of the
style representation for all frames are plotted in Figure \ref{convergence} B,
which shows that the style representation obeys a Gaussian distribution. From
the visualization of frames in every cluster, we found that the algorithm
identified several featured conformations, for example, cluster 0, 5, 6 and 7
being identified as well-packed conformations, cluster 4 as the unfolded
conformation and cluster 3 as some interesting intermediate state on its way to
a fully fold conformation. We notice that the frames in cluster 0, 5, and 6 are
quite similar and could be treated as duplicate clusters. Intuitively this makes
sense, because the folded structure has a greater population among all states
since, in the conditions of the simulation, it is energetically stable but we
forced every cluster contain the same number of frames.

\subsection{Clustering with AAE-GS}

Using the same assumption about the distribution of the style vector and its
dimension as in the previous experiment, we did a clustering experiment using
AAE with Gumbel-Softmax. We use a similar architecture as the previous
experiment except that we included the reparametrization with Gumbel-Softmax
(See Supplementary Info C). Based on the model of AAE-GS, we define an annealing
process of the parameter $\tau$ (See Figure 1 in Supplementary Info) to cool
down the categorical representation from a uniform vector to a one-hot vector.
We trained the model until convergence. AAE-GS takes similar amount of time as
the AAE experiment using a similar model architecture. The convergence of the
reconstruction loss and the successful imposition of the Gaussian noise are
shown in Figure \ref{convergence} A and C.

\begin{figure}[h]
  \centering
  \includegraphics[width=5in]{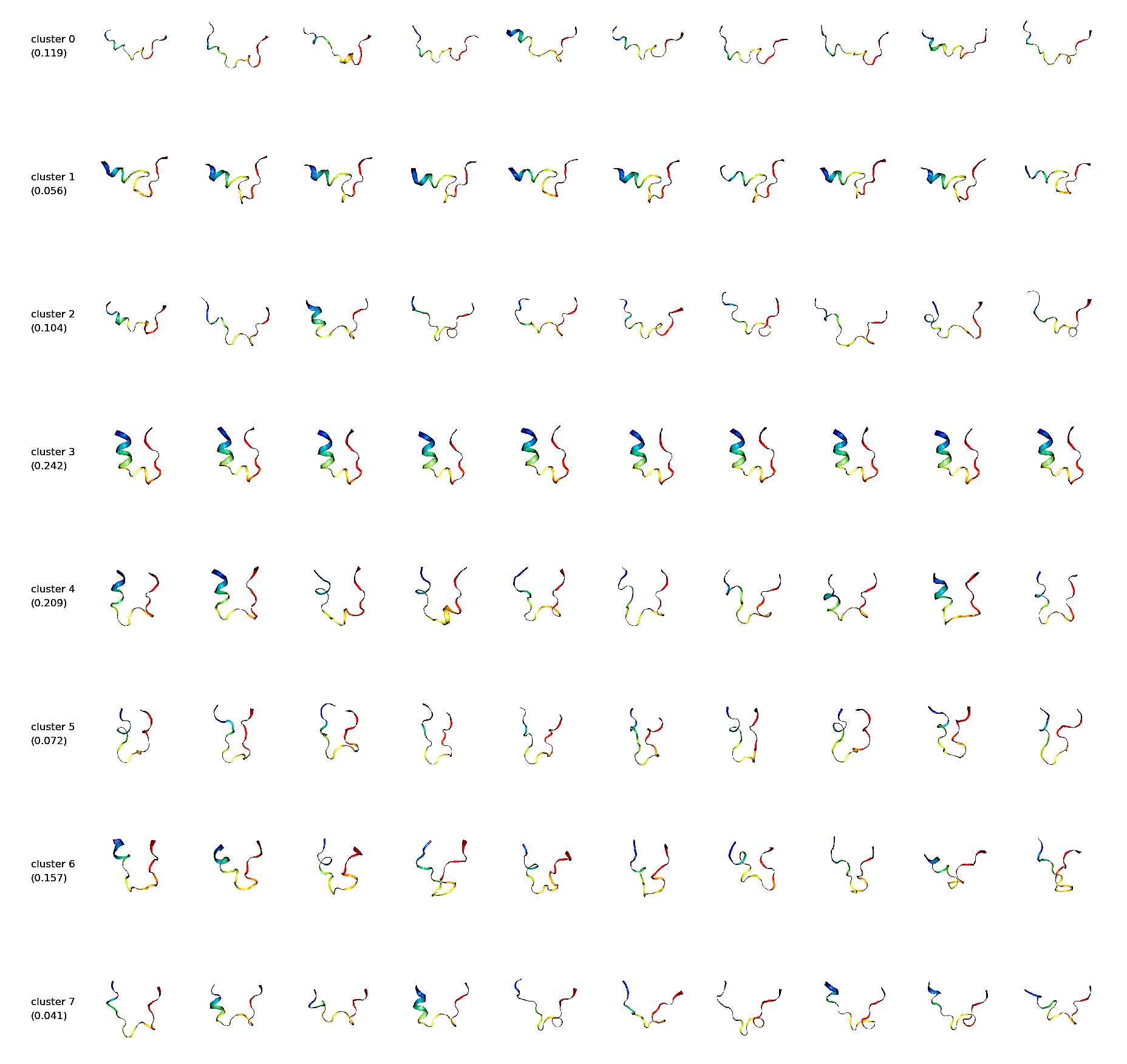}
  \caption{Visualization of cluster representatives obtained from AAE-GS experiment.}
  \label{aae_gs_clustering_result}
\end{figure}

This experiment shows that the folded structures are clustered into only one
cluster (cluster 3) instead of three duplicated clusters obtained in the AAE
experiment. Cluster 3 contains frames with the same characteristics of the
folded structure: Trp6 is caged by residues Pro12, Pro18, Tyr3, Leu7 and Pro19.
(See Figure \ref{fig:structural}) The probability mass for this folded cluster
is approximately twice the mass of the same cluster in the AAE experiment. It
seems that the AAE-GS algorithm tends to aggregate all folded structures within
those duplicate clusters into one. However, the probability mass was doubled,
not tripled, because the structures in a AAE cluster may not be strictly folded
as the norm of the style vector gets larger. The algorithm can categorize those
loosely packed structures into other clusters. By comparing all 8 clusters in
Figure \ref{aae_gs_clustering_result}, we found there are no obvious duplicated
clusters.

\subsection{Generation of fake frames}

As explained in the Theory section, with a trained decoding function and latent
variables with known distributions, we can generate fake frames that do not
exist in the original trajectory. Figure \ref{fakeframes} shows a few fake frame
examples generated from cluster 4 using the AAE-GS model. The sampled style
vector used to generate the fake frame is shown below the structure.

\begin{figure}[h]
  \centering
  \includegraphics[width=4in]{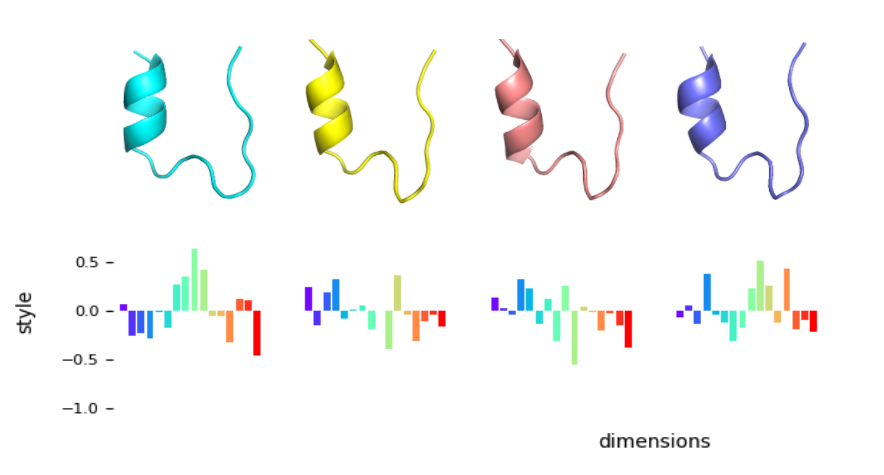}
  \caption{Examples of fake frames generated from the AAE-GS model. The
  corresponding Gaussian noise vectors are plotted as bar plots. All 16 elements
  of the style representation are spread out along the horizontal
  axis of each bar plot. Different dimensions are plotted in different colors.}
  \label{fakeframes}
\end{figure}

The reconstruction quality can be furthered improved by optimizing the
hyperparameters of the model components. For instance, we could add Convolution
Neural Networks (CNN) to our model to increase the representation power of
all components, which has proven to be very effective in computer vision tasks.
Generation of fake frames can enrich the existing dynamics and could be useful
for searching for new conformations with specific biophysical relevance.

\subsection{An application}

To evaluate the significance of the clusters identified in the long Trp-Cage MD
trajectory by our AAE-GS computation we assumed that they represent
significant intermediates in the folding of the small protein. Several groups
have studied the folding of Trp-Cage using MD simulations
\cite{trpcage1}\cite{trpcage2}\cite{trpcage3}\cite{trpcage4}\cite{trpcage5}.
Some of these studies used clustering and other algorithms to infer a possible
folding pathway. In our clustering, it is clear that cluster 0 corresponds to
the unfolded state (U) and that cluster 3 corresponds to the fully folded
protein (F). The significance of these and the other clusters was explored by
using all the clusters as the states of a Markov Model (MM). This computation
provided the probabilities of transition between the different states, estimated
as
\begin{equation}
p(s_i, s_j) = \frac{c(s_i, s_j)}{\sum_k\sum_l c(s_k, s_l)}
\end{equation}
where $c(s_i, s_j)$ represents the count of the transitions from
state $i$ to state $j$. It is worth mentioning that $p(s_i, s_j)$ is not a joint
probability since $i$ and $j$ are in order. The transition probabilities are
defined in such a way that the proportions of each state is reflected.

As seen in Figure \ref{transitions}, some states have only a single probability
of transition (besides the return to the same state) while others have
transitions to more than one additional state. Based on their probabilities of
transition the states can be connected in a way that reflects possible pathways
from U (state 0) to F (state 3). In this scheme (Figure \ref{transitions}),
state 2 is a required intermediate in the folding process.  In this state, the
helix is starting to form.  The two prolines (Pro12 and Pro18) are pointing
towards each other and, with Tyr3 they are starting to form the "cage".  The
tryptophan (Trp6) is still outside the cage but as the helix continues to grow
past residue 5, the tryptophan will start to become the center of the cage. It
can reach the folded state directly but it can also transition to states 1, 4 or
7. State 1 is a dead-end that has to return to state 2 to reach the folded state
3.  State 7 can reach the folded state through two paths: returning to state 2
or transition to state 4 that in turn can reach the folded state 3.  State 4 has
the largest number of significant transition probabilities: to states 2, 7, 5
and to the folded state (3). It is the last step before reaching the final fold.
It has most residues in the correct conformation and would require only an
\textasciitilde $100^{\circ}$ counterclockwise rotation of the helix and small
rearrangement to reach the folded state. State 5 is another dead-end. It
apprears that in state 5, the protein is attempting to fold as an antiparallel
hairpin. Since this arrangement does not lead to favorable interactions nor can
it evolve into the final fold.  It has to return to state 4 to reach the folded
state. State 6 is interesting in that it is a not fully-folded conformation that
is only accessible from the folded state and may represent a partially unfolded
state in equilibrium with the folded state in the conditions of the simulation.
This state 4 has a high probability of transition to itself, second only to that
of the folded state. The structural descriptions for states mentions above are
reflected in Figure \ref{fig:structural}.

\begin{figure}
  \centering
  \includegraphics[width=5in]{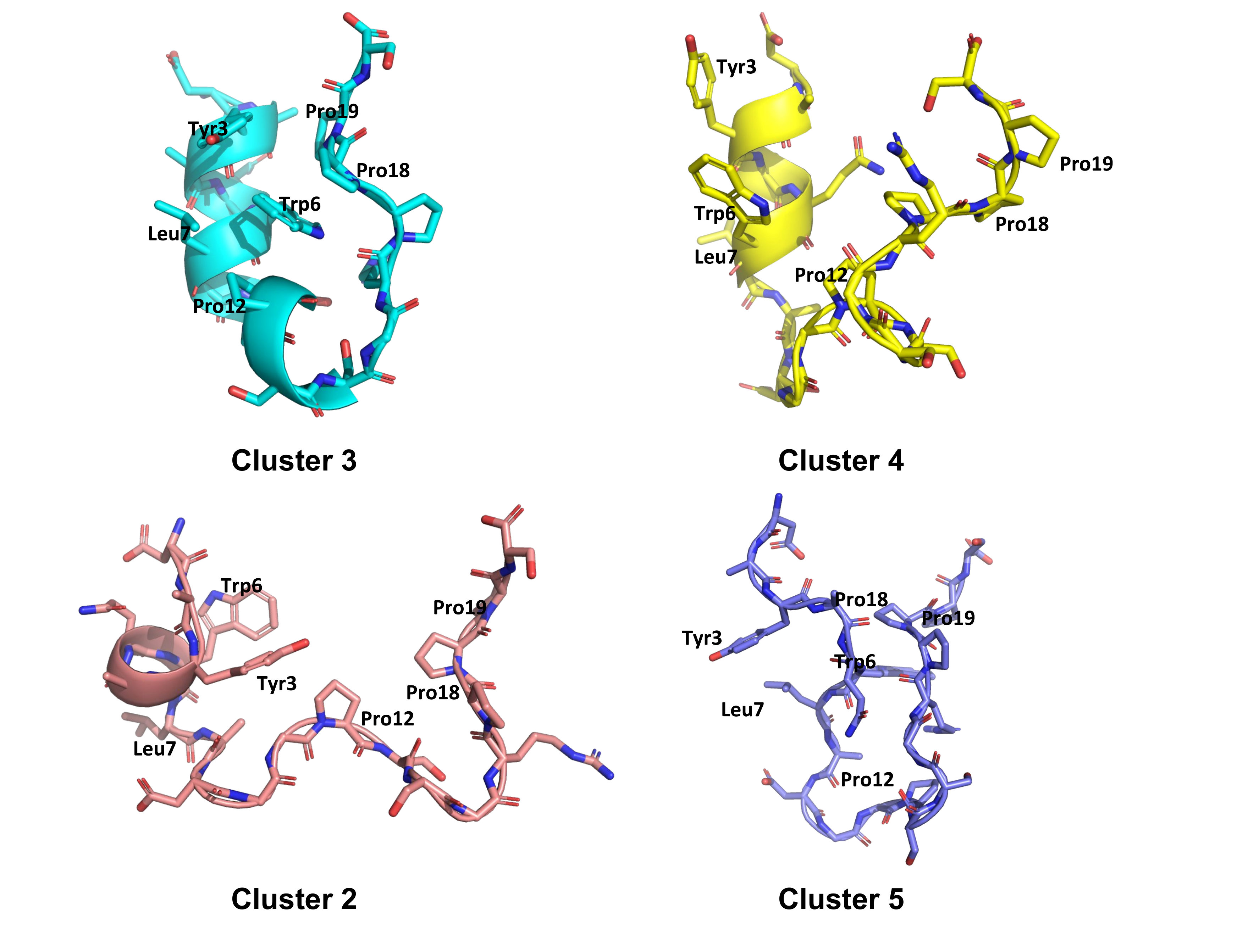}
  \caption{Structural descriptions for selective states. The backbone of the
  protein is plotted with the representation Cartoon and key residues are shown
  with balls and sticks.}
  \label{fig:structural}
\end{figure}

The combination of AAE-GS-identified clusters with a MM using
the clusters as the model states, represents an objective analysis of an MD
trajectory that provides a highly interpretable model of the folding pathway
described by the simulation. The same combination of MD simulations, AAE-GS
clustering and MM can be used to identify transitions in enzymes, transporters,
channels and others proteins and can become an unbiased procedure to gain
insight about the mechanism these proteins.

\begin{figure}[h]
  \centering
  \includegraphics[width=5in]{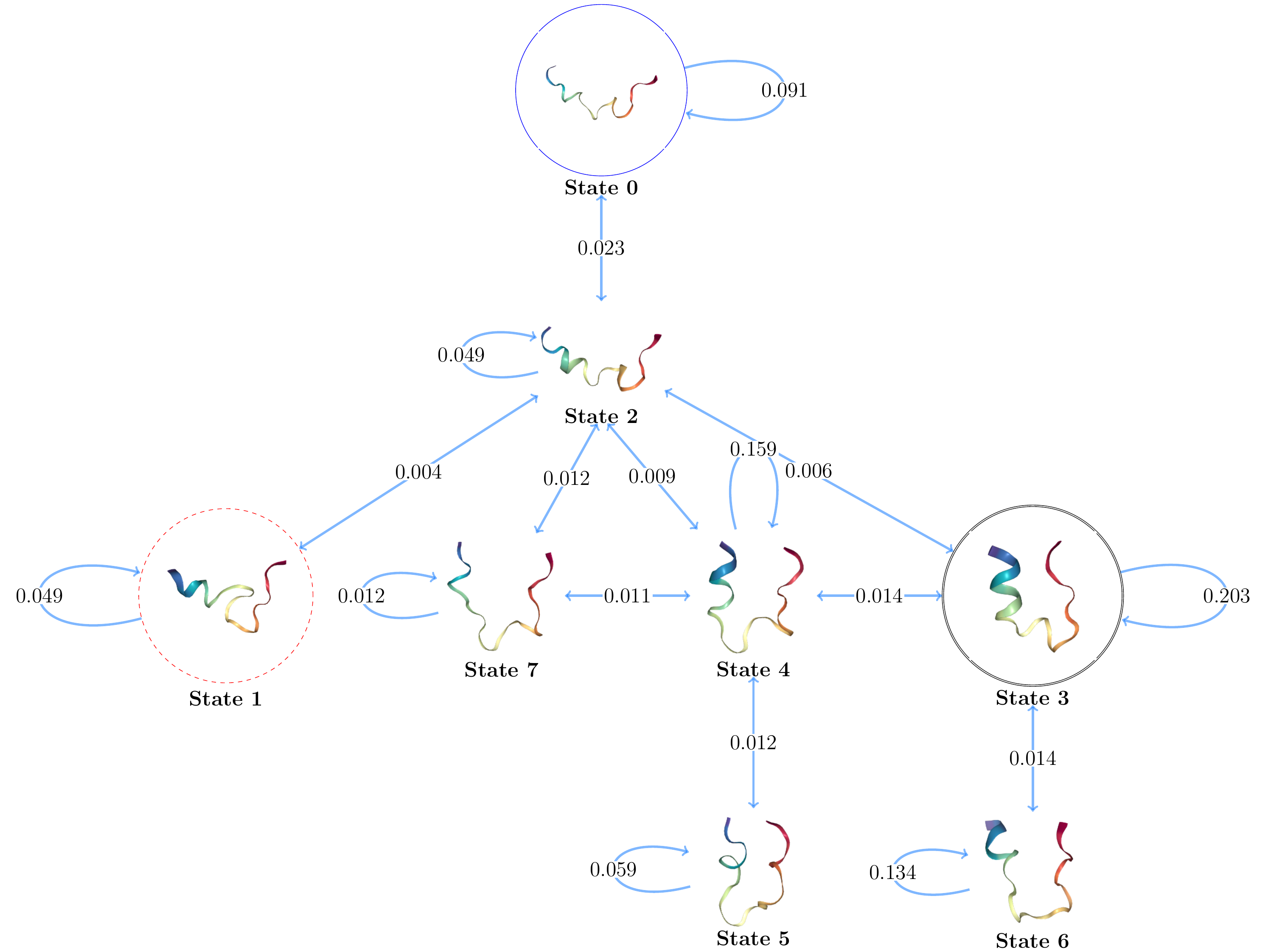}
  \caption{Transitions between different states of Trp-cage. The first frame
  in each cluster is used as the representation of the cluster. The unfolded
  state is circled in blue. The perfect folded state is double circled in black.
  A misfolded state (State 1) is circled in red dashes. All transition probabilities
  are sit in the middle of the arrow. Arrows with a transition probability
  lower than 0.005 are pruned.}
  \label{transitions}
\end{figure}

\section{Discussion and Conclusion}

Unlike clustering images, clustering MD frames is intrinscally hard since the
frame set, which usually comes from trajectories of MD simulation, contains all
intermediate conformations of a continuous transformation from one featured
state of the macromolecule to another. Clustering such continuous paths is
equivalent to determining boundaries at some points of the paths. Consequently, assigning the
frames within some neighborhood of the boundaries is somewhat arbitrary. In more formal
terms, as our decoder is actually a continuous function, frames with smaller
$|\mathbf{z}|$ reveal the specific features of a cluster better than those with
larger $|\mathbf{z}|$. For those frames with large $|\mathbf{z}|$, the algorithm
is less accurate about their assignment. This indetermination could be resolved
by having more clusters, however, more clusters will make features attributed to
each cluster less distinguishable. Figures 3, 4 in Supplementary Info plot a
random selection of frames from the first 10,000 frames of each cluster with the
same clustering results from the AAE and AAE-GS experiments, respectively. We
found that the majority of frames in the random selection still preserve
specific features that identified their clusters.

From a technical perspective, the general architecture we propose for clustering
could be enhanced by involving CNNs and other types of networks. The
hyperparameters of all network components are subject to automatic tuning in
order to achieve the smallest reconstruction error. Currently, a major caveat of
training our model is the lack of robustness of the two GAN components. However,
as more and more efforts\cite{WGAN}\cite{WGAN_GP} on stablizing GAN are brought
into play in the machine learning community, we expect to have a stable way of
training GAN components and tuning the hyperparameters of the discriminators.

We showed that the AAE-GS algorithm we propose is effective in clustering MD
accessed conformations and it naturally provides a quantity $|\mathbf{z}|$ that
represents the loyalty of a frame to the cluster it belongs to. In addition, we
showed that our clustering results can be useful in revealing and quantifying
important biophysical properties of macromolecules, such as folding landscapes,
allosteric rearrangements and transitions among active and inactive
conformations.

\subsubsection*{Acknowledgments}

We acknowledge Maryland Advanced Research Computing Center (MARCC) for
providing hardware resource and support. We thank D.E. Shaw Research for
kindly providing us the 208 $\mu$s trajectory of Trp-Cage.

\clearpage
\begin{center}
\textbf{\large Supplemental Info: Conformation Clustering of Long MD Protein
  Dynamics with an Adversarial Autoencoder}
\end{center}

\setcounter{equation}{0}
\setcounter{figure}{0}
\setcounter{table}{0}
\setcounter{page}{1}
\setcounter{section}{0}
\makeatletter
\renewcommand{\theequation}{S\arabic{equation}}
\renewcommand{\thefigure}{S\arabic{figure}}
\renewcommand{\bibnumfmt}[1]{[S#1]}
\renewcommand{\citenumfont}[1]{S#1}
\renewcommand\thesection{\Alph{section}}

\section{Proof of Proposition 2.1}

\begin{lemma}\label{l1}
Let $\mathbf{g}$ be a $d$-dimensional random vector and all its components are independent
identical distributed (i.i.d) random variables
with a $\text{Gumbel }(0, 1)$ distribution.
$\boldsymbol{\pi}$ be a $d$-dimensional vector in $(0, +\infty)^d$ and $\tau$ is an arbitrary
number in $R_{>0}$. Let's define a $d$-dimensional random vector $\mathbf{y}$ by transforming
$\mathbf{g}$ with function $\mathbf{f}$, where $\mathbf{f}$: $R^d \rightarrow R^d$.
The $k$th element of $\mathbf{f}$ is
\begin{equation}\label{sample}
y_k = f_k(\mathbf{g}) = \frac{\exp{((\log{(\pi_k)} + g_k) / \tau)}}
{\sum_{k} \exp{((\log{(\pi_k)} + g_k) / \tau)}}
\end{equation}
where $y_k$, $g_k$ and $\pi_k$ are the $k$th element of $\mathbf{y}$, $\mathbf{g}$ and
$\boldsymbol{\pi}$, respectively. Then we have, $\forall k' \in \left\{1, \dots, d\right\}$,

\begin{equation}
p(y_{k'} = \max_k y_k \mid \boldsymbol{\pi}, \tau) = \frac{\pi_{k'}}{\sum_k {\pi_k}}
\end{equation}
\end{lemma}

\begin{proof}
Let's first consider the conditional probability of $y_{k'} = \max_k y_k$ given $\boldsymbol{\pi}$, $\tau$
and $g_{k'} = g$. Since $g_k$s are i.i.d. and $y_{k'}$ is the maximum among all elements of $\mathbf{y}$,
we have,

\begin{equation}
\begin{split}
p(y_{k'} = \max_k y_k \mid \boldsymbol{\pi}, \tau, g_{k'}=g) & =
\prod_{k \neq k'} p(y_k \leq y_{k'}) \\
 & = \prod_{k \neq k'} p(\log{\pi_k} + g_k \leq \log{\pi_{k'}} + g) \\
 & = \prod_{k \neq k'} p(g_k \leq \frac{\pi_{k'}}{\pi_k} + g)
\end{split}
\end{equation}

Since $g_k \sim \text{Gumbel}(0, 1)$ and the cumulative distribution function of
Gumbel(0, 1) is $e^{-e^{-x}}$, $x \in \mathbb{R}$, we substitute
$p(g_k < \frac{\pi_{k'}}{\pi_k} + g)$ with $\exp{(-\exp{(-\log{(\frac{\pi_{k'}}{\pi_k})} - g)})}$.
Then we have,

\begin{equation}
\begin{split}
p(y_{k'} = \max_k y_k \mid \boldsymbol{\pi}, \tau, g_{k'}=g) & =
\prod_{k \neq k'} \exp{(-\exp{(-\log{(\frac{\pi_{k'}}{\pi_k})} - g)})} \\
 & = \exp{\sum_{k \neq k'} (-\exp{(-\log{(\frac{\pi_{k'}}{\pi_k})} - g)})} \\
 & = \exp{\sum_{k \neq k'} (-\frac{\pi_{k}}{\pi_{k'}} e^{-g})} \\
 & = \exp{(-e^{-g} \frac{\sum_{k \neq k'} \pi_k}{\pi_{k'}})}
\end{split}
\end{equation}

The probability density function of Gumbel(0, 1) is $e^{-x-e^{-x}}$, where $x \in \mathbb{R}$.
Marginalize out the $g_{k'}$, we have,

\begin{equation}\label{marginalize}
\begin{split}
p(y_{k'} = \max_k y_k \mid \boldsymbol{\pi}, \tau) & =
\int_{g_{k'}} p(y_{k'} = \max_k y_k \mid \boldsymbol{\pi}, \tau, g_{k'}=g) p_{g_{k'}}(g) dg \\
 & = \int_{-\infty}^{+\infty} \exp{(-e^{-g} \frac{\sum_{k \neq k'} \pi_k}{\pi_{k'}})} e^{-g-e^{-g}} dg \\
 & = \int_{-\infty}^{+\infty} \exp{(-g-e^{-g}\frac{\sum_k{\pi_k}}{\pi_{k'}})} dg
\end{split}
\end{equation}

Let $C = \frac{\sum_k{\pi_k}}{\pi_{k'}}$, the improper integral in the last step
of Equation \ref{marginalize} can be rewritten as
$\int_{-\infty}^{+\infty} \exp{(-g-Ce^{-g})} dg$.

This integral can be easily calculated by substituting $g$ with $-\log{t}$. We have,
\begin{equation}
\begin{split}
\int_{-\infty}^{+\infty} \exp{(-g-Ce^{-g})} dg & = \int_{0}^{+\infty} e^{-Ct} dt =  \\
 & = \Eval{-\frac{1}{C}e^{-ct}}{0}{+\infty} \\
 & = \frac{1}{C}
\end{split}
\end{equation}

Subsituting $C$ back to $\frac{\sum_k{\pi_k}}{\pi_{k'}}$ completes the proof.
\end{proof}

\begin{prop}
Let $\boldsymbol{y}$ be a random variable sampled from Equation \ref{sample}, where
$\boldsymbol{\pi} \in (0, +\infty)^d$, $\tau \in (0, +\infty)$, then we have, \\
\begin{enumerate}[(1)]
  \item $\forall k$, $\lim\limits_{\tau \rightarrow +\infty} y_k = \frac{1}{d}$. In fact,
  $\lim\limits_{\tau \rightarrow +\infty} y_k$ defines a random variable, which degenerates to a
  constant $\frac{1}{d}$ everywhere in the event space $\Omega$, which is $\mathbb{R}^k$.
  \item $\lim\limits_{\tau \rightarrow 0^{+}} \mathbf{y} \sim \text{OneHotCategorical}(\frac{\boldsymbol{\pi}}{\sum_{i=1}^{k}\pi_i})$
  \item $\tau \rightarrow 0^{+}$, $\mathbb{E}{y_i} = \frac{\pi_i}{\sum_{j=1}^{k}\pi_j}$.
\end{enumerate}
\end{prop}

\begin{proof}
\begin{enumerate}[(1)]
\item
$\forall \omega \in \Omega$, $\forall k \in \left\{1, \dots, d\right\}$,
$|g_k| < +\infty$, so as $|\log{(\pi_k)} + g_k| < +\infty$. In this case, when $\tau \rightarrow +\infty$,
$(\log{(\pi_k)} + g_k) / \tau \rightarrow 0$ and $\exp{((\log{(\pi_k)} + g_k) / \tau)} \rightarrow 1$.
It is trivial to see $y_k = \frac{1}{d}$.

\item

Since $\mathbf{g}$ is a random vector defined on the event space $\mathbb{R}^d$,
which is denoted as $\Omega$ in the following text, $\mathbf{y} = \mathbf{f}
\circ \mathbf{g}$ is also defined on $\Omega$. Let's start by defining a
partition of $\Omega$. \\

Consider a series of subsets of $\Omega$, denoted by $\big\{A_k\big\}$ where $k$
is the index.  Each $A_k$ is defined as $\big\{\omega \mid \omega \in \Omega, k
\in I_{max}, |I_{max}| = 1\big\}$, where $I_{max} = \argmax_{k} (g_k(\omega) +
\log{(\pi_k)})$ and $|\cdot|$ denotes the cardinality. Let's define another set
$A_{>1}$ as $\big\{\omega \mid \omega \in \Omega, |I_{max}| > 1\big\}$. Since
$\mathbf{g}$ is a $d$-dimensional random vector, $\forall \omega$, $|I_{max}|
\leq d$. In addition, since $\forall \omega$, each component of
$\mathbf{g}(\omega)$ is finite and bounded, then there must be at least one
maximum. This implies $|I_{max}| \geq 1$, $\forall \omega$. Therefore, it is
trivial to see $(\cup_k A_k) \cup A_{>1}$ is a partition of $\Omega$. \\

Let's consider the value of $\lim\limits_{\tau \rightarrow 0^{+}} y_k$ on each
$A_k$. $\forall \omega \in A_k$, by definition, we have, for $\forall j$, $j \in
\big\{1,\dots,d \big\}$,

\begin{equation}\label{divide}
\begin{split}
y_{j} = f_{j}(\mathbf{g}(\omega)) & = \frac{\exp{((\log{(\pi_{j})} + g_{j}) / \tau)}}
{\sum_{i} \exp{((\log{(\pi_i)} + g_i) / \tau)}} \\
& = \frac{\exp{((\log{(\pi_j)} + g_j - \log{(\pi_{k}) -g_{k}}) / \tau)}}
{\sum_{i} \exp{((\log{(\pi_i)} + g_i - \log{(\pi_{k}) -g_{k}}) / \tau)}}
\end{split}
\end{equation}

The last step of Equation \ref{divide} is obtained by dividing the term
$\exp{(\log{(\pi_k)} + g_k))}$ from both the numerator and the denominator. It is trivial
to see that for any $j \neq k$,
$\lim\limits_{\tau \rightarrow 0^{+}} \exp{((\log{(\pi_j)} + g_j - \log{(\pi_{k}) -g_{k}}) / \tau)} = 0$
and only when $j = k$, the limit is equal to 1. This is also true for $i$ in the denominator. Then
we can conclude, for any index $k$,

\begin{equation}\label{imax1}
\forall \omega \in A_k, \lim\limits_{\tau \rightarrow 0^{+}} y_j = 1 \text{ when } j = k,
\text{ otherwise } 0.
\end{equation}

Similarly as the above analysis, we can draw the corresponding conclusion for set $A_{>1}$.

\begin{equation}\label{imax2}
\forall \omega \in A_{>1}, \lim\limits_{\tau \rightarrow 0^{+}} y_j = \frac{1}{|I_{max}|} \text{ when } j \in I_{max}, \text{ otherwise } 0.
\end{equation}

For simplicity, we use $\mathbf{e}_k$ to denote the unit vector on $k$-th
dimension in $\mathbb{R}^d$, which is defined as only the $k$-th element is 1
and all others are 0. And we denote the set that contains all values of
$\lim\limits_{\tau \rightarrow 0^{+}} \mathbf{y}$ in Equation \ref{imax2} as
$O$. Then Equation \ref{imax1} implies $p(A_k) \leq p(\lim\limits_{\tau
\rightarrow 0^{+}} \mathbf{y} = \mathbf{e}_k)$ and Equation \ref{imax2} implies
$p(A_{>1}) \leq p(O)$. \\

It is also trivial to see that $\lim\limits_{\tau \rightarrow 0^{+}}
\mathbf{y}(\omega) = \mathbf{e}_k$ implies $\omega \in A_k$. This can be proved
by contradiction. Assume $\omega \notin A_k$, since $(\cup_k A_k) \cup A_{>1}$
is a partition of $\Omega$, $\omega$ must be in any other set $A_j$ where $j
\neq k$ or $A_{>1}$. Apparently this is impossible because of Equation
\ref{imax1} and \ref{imax2}. Then we have $p(\lim\limits_{\tau \rightarrow 0^{+}}
\mathbf{y} = \mathbf{e}_k) \leq p(A_k)$ and $p(O) \leq p(A_{>1})$.

\begin{equation}\label{result}
\begin{split}
p(A_k) &= p(\lim\limits_{\tau \rightarrow 0^{+}} \mathbf{y} = \mathbf{e}_k) \\
p(A_{>1}) &= p(O) \\
\end{split}
\end{equation}

Due to Lemma \ref{l1} and its proof, we know $p(A_k) = \frac{\pi_{k}}{\sum_i
{\pi_i}}$. (Please notice that Gumbel(0, 1) is a continuous distribution, so
$p(y_k \leq g_{k'})  = p(y_k < g_{k'})$ in the proof of Lemma \ref{l1}.) By
summing $k$, we have $\sum_k p(A_k) = 1$. On the other hand, because of the fact
that $(\cup_k A_k) \cup A_{>1}$ is a partition of $\Omega$, $\sum_k p(A_k) +
p(A_{>1}) = 1$. Then we get $p(A_{>1}) = 0$, which makes the points in $A_{>1}$
and $O$ unmeaningful. (This can be seen in another way
that $A_{>1}$ is exactly a zero-measure set.) It turns out that $\lim\limits_{\tau \rightarrow 0^{+}} \mathbf{y}$ is defined only on $\big\{\mathbf{e}_k\big\}$. Take $\big\{\mathbf{e}_k\big\}$
as $\Omega'$ and $\mathcal{F}'$ is the $\sigma$-algebra on it. Combining with the result in
Equation \ref{result}, we get the distribution of $\lim\limits_{\tau \rightarrow 0^{+}}
\mathbf{y}$ as,

\begin{equation}
p'(\lim\limits_{\tau \rightarrow 0^{+}} \mathbf{y} = \mathbf{e}_k) = \frac{\pi_{k}}{\sum_i
{\pi_i}}
\end{equation}
which is usually written as

\begin{equation}
p'(\lim\limits_{\tau \rightarrow 0^{+}} \mathbf{y}) = \prod_{j} (\frac{\pi_{j}}{\sum_i
{\pi_i}})^{y_j}
\end{equation}
where $p'$ is the probability measure defined on the probability space $(\Omega', \mathcal{F}', p')$.

\item
This is a trivial corollary of (2).

\end{enumerate}
\end{proof}

\section{AAE Experiment}
In our experiment, the architecture of AAE we used consists of 8 multilayer perceptrons (MLP):
1) The Shared Encoder encoding the raw data to a shared deep representation that could be further decomposed into a categorical id and the noise, 2) The Categorical Encoder mapping the shared representation to
an one-hot categorical vector, 3) The Style Encoder translating the shared representation to a
gaussian noise, 4) The Categorical Decoder and 5) The Style Decoder reconstructing the shared representation
from the categorical id and the noise, 6) The Shared Decode reconstructing the raw data from the shared
representation, 7) The Discriminator of Categorical-WGAN and 8) The Discriminator of Style-WGAN. The
structure of each components are listed in the following table \ref{AAEModel}.

\begin{table}[]
\centering
\caption{Structures of all sub-components of our AAE model.}
\label{AAEModel}
\begin{tabular}{|c|c|}
 \hline
 Component & Architecture\footnote{Since all components are modeled with MLP,
 we represents the layers of MLP with numbers separate with "$\times$". Each number refers
 to the number of nodes in that layer.} \\
 \hline
 Shared Encoder & $400_{\text{BN, 0.1 DO}}\times400_{\text{BN}}\times400_{\text{BN}}$
 \footnote{Note for the subscript and superscript: "BN" refers to Batch Normalization; "DO" prefixed with a number refers to dropout and the number is the dropout probability; Superscript refers to an activation
 function other than "relu". We use "relu" as the default activation function.} \\
 \hline
 Categorical Encoder & $100_{\text{BN}}\times8^{\text{softmax}}$ \\
 \hline
 Style Encoder & $200_{\text{BN}}\times16^{\text{id}}$ \\
 \hline
 Categorical Decoder & $100_{\text{0.1 DO}}\times400_{\text{BN}}$ \\
 \hline
 Style Decoder & $200_{\text{0.1 DO}}\times400_{\text{BN}}$ \\
 \hline
 Shared Decoder & $400_{\text{BN, 0.1 DO}}\times400_{\text{BN}}\times240_{\text{BN}}^{\text{id}}$ \\
 \hline
 Discriminator (Cat-GAN) & $400_{\text{0.1 DO}}\times400_{\text{0.1 DO}}\times400\times200\times200\times200\times200\times1^{\text{id}}$\\
 \hline
 Discriminator (Style-GAN) & $400_{\text{0.1 DO}}\times400_{\text{0.1 DO}}\times400\times200\times200\times200\times200\times1^{\text{id}}$\\
 \hline
\end{tabular}
\end{table}

We use the Adam Optimizer ($\beta_1 = 0.5$, $\beta_2 = 0.9$) to optimize all
loss functions with a minibatch of 4096 shuffled examples. The learning rate for
updating the autoencoder is $2.0 \times 10^{-5}$. For both the categorical GAN
and the style GAN, we use a learning rate of $2.0 \times 10^{-4}$ to train and a
gradient penalty weight of $10$ to regularize the weights.

\section{AAE with Gumbel Softmax Experiment}
We use a very similar model architecture as the AAE Experiment. We simply replace
the activation function of the last layer of Categorical Encoder with "id". This
is because we need to produce "fake" samples of the Gumbel distributions. \\

We anneal the parameter $\tau$ from 10 to 0.01 in the first 20000 steps with a
power 4 of polynomial decay. After 20000 steps, $\tau$ was remained as 0.01.
The following figure depicts the annealing process of $\tau$.

\begin{figure}[h]
  \centering
  \includegraphics[width=3.5in]{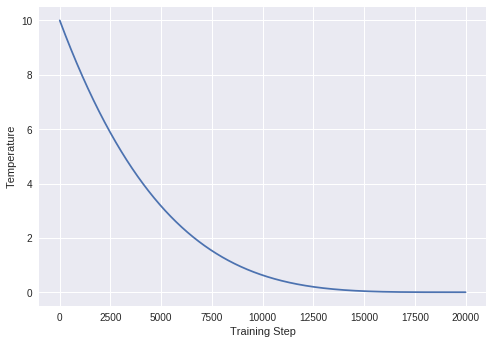}
  \caption{The annealing process of the parameter $\tau$.}
\end{figure}

\clearpage
\section{Architecture of AAE}
\begin{figure}[h]
  \centering
  \includegraphics[width=5in]{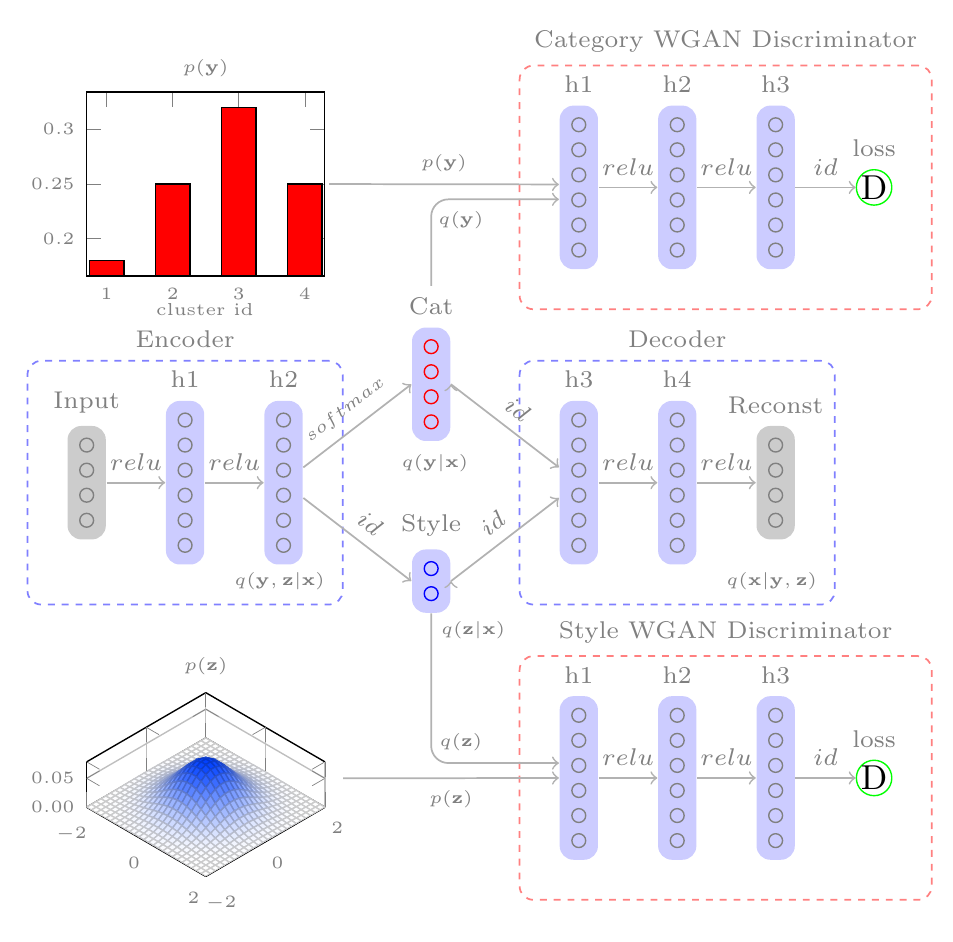}
  \caption{The general architecture of AAE for clustering. This architecture is represented
  in a similar way as the one of AAE with Gumbel-Softmax.}
\end{figure}

\clearpage
\section{Visualization on Random Selections of Clustering Results}
\begin{figure}[h]
  \centering
  \includegraphics[width=5in]{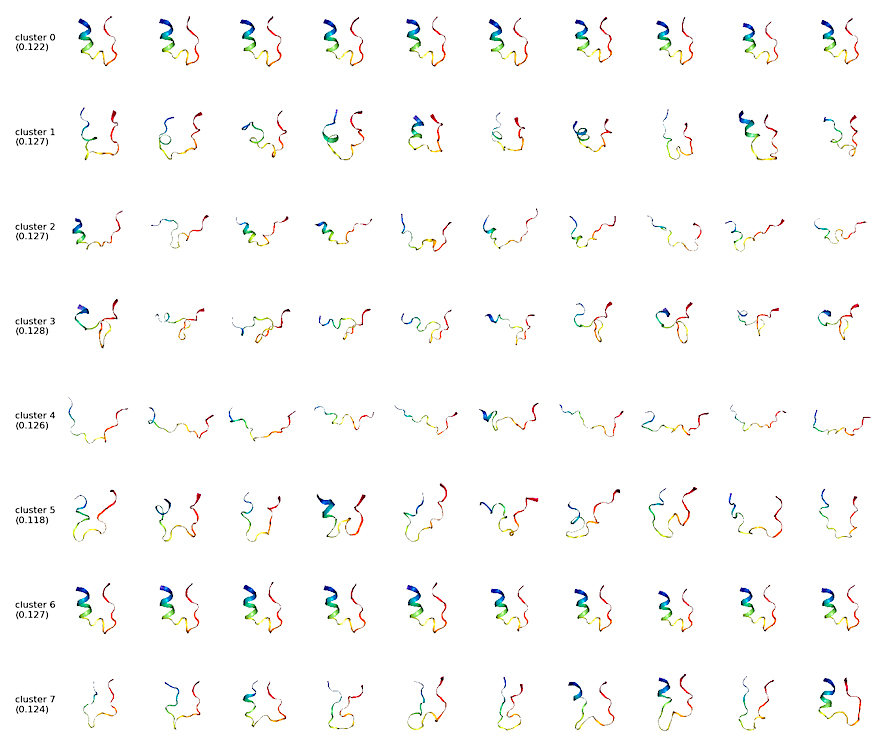}
  \caption{Random frames selected from the first 10,000 of each cluster (AAE)}
\end{figure}

\begin{figure}[h]
  \centering
  \includegraphics[width=5in]{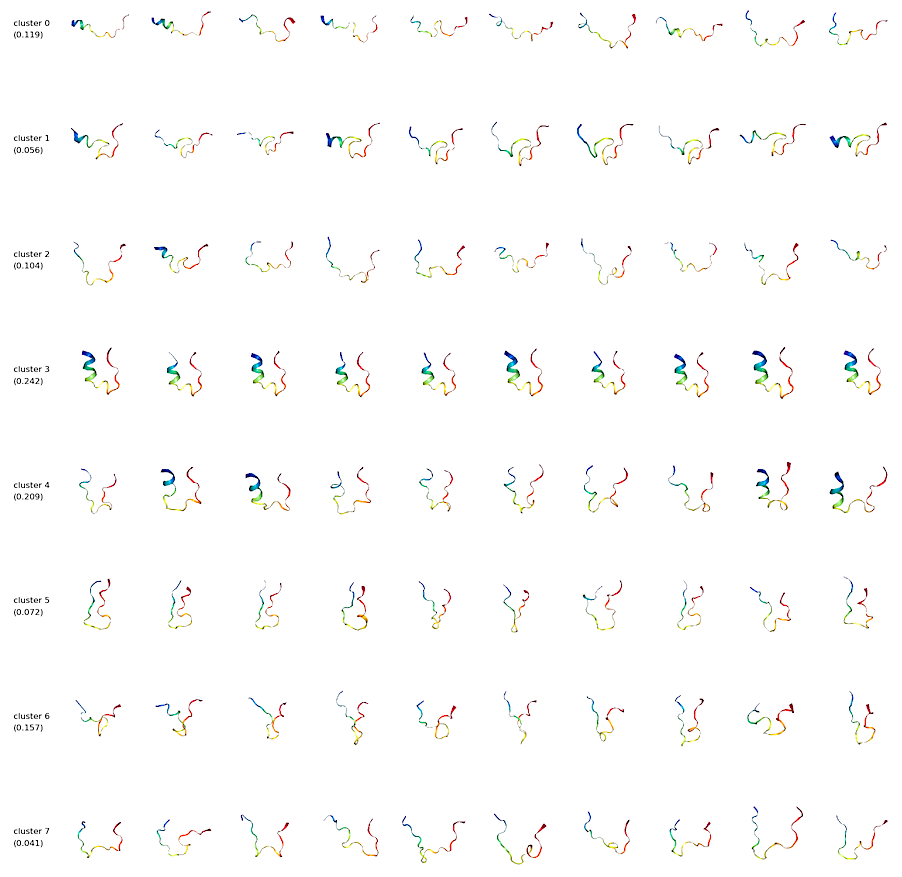}
  \caption{Random frames selected from the first 10,000 of each cluster (AAE-GS)}
\end{figure}


\clearpage

\begin{thebibliography}{10}

\bibitem{Tensorflow}
Martin Abadi, Paul Barham, Jianmin Chen, Zhifeng Chen, Andy Davis, Jeffrey
  Dean, Matthieu Devin, Sanjay Ghemawat, Geoffrey Irving, Michael Isard,
  Manjunath Kudlur, Josh Levenberg, Rajat Monga, Sherry Moore, Derek~G. Murray,
  Benoit Steiner, Paul Tucker, Vijay Vasudevan, Pete Warden, Martin Wicke, Yuan
  Yu, and Xiaoqiang Zheng.
\newblock Tensorflow: A system for large-scale machine learning.
\newblock In {\em 12th USENIX Symposium on Operating Systems Design and
  Implementation (OSDI 16)}, pages 265--283, 2016.

\bibitem{Gromacs}
Mark~James Abraham, Teemu Murtola, Roland Schulz, Szilárd Páll, Jeremy~C.
  Smith, Berk Hess, and Erik Lindahl.
\newblock Gromacs: High performance molecular simulations through multi-level
  parallelism from laptops to supercomputers.
\newblock {\em SoftwareX}, 1–2:19 -- 25, 2015.

\bibitem{Metrics}
Charu~C. Aggarwal, Alexander Hinneburg, and Daniel~A. Keim.
\newblock On the surprising behavior of distance metrics in high dimensional
  space.
\newblock In {\em Lecture Notes in Computer Science}, pages 420--434. Springer,
  2001.

\bibitem{WGAN}
Martin Arjovsky, Soumith Chintala, and L{\'e}on Bottou.
\newblock {W}asserstein generative adversarial networks.
\newblock In Doina Precup and Yee~Whye Teh, editors, {\em Proceedings of the
  34th International Conference on Machine Learning}, volume~70 of {\em
  Proceedings of Machine Learning Research}, pages 214--223, International
  Convention Centre, Sydney, Australia, 06--11 Aug 2017. PMLR.

\bibitem{charmm}
B.R. Brooks, III C.L.~Brooks, and Jr. A.D.~MacKerell.
\newblock Charmm: The biomolecular simulation program.
\newblock {\em J Comput Chem.}, 30(10):1545--1614, 2009.

\bibitem{Gromos}
Xavier Daura, Karl Gademann, Bernhard Jaun, Dieter Seebach, Wilfred~F van
  Gunsteren, and Alan~E Mark.
\newblock Peptide folding: when simulation meets experiment.
\newblock {\em Angewandte Chemie International Edition}, 38(1-2):236--240,
  1999.

\bibitem{GAN}
Ian Goodfellow, Jean Pouget-Abadie, Mehdi Mirza, Bing Xu, David Warde-Farley,
  Sherjil Ozair, Aaron Courville, and Yoshua Bengio.
\newblock Generative adversarial nets.
\newblock In Z.~Ghahramani, M.~Welling, C.~Cortes, N.~D. Lawrence, and K.~Q.
  Weinberger, editors, {\em Advances in Neural Information Processing Systems
  27}, pages 2672--2680. Curran Associates, Inc., 2014.

\bibitem{WGAN_GP}
Ishaan Gulrajani, Faruk Ahmed, Mart{\'{\i}}n Arjovsky, Vincent Dumoulin, and
  Aaron~C. Courville.
\newblock Improved training of wasserstein gans.
\newblock {\em CoRR}, abs/1704.00028, 2017.

\bibitem{Gumbel}
E.J. Gumbel.
\newblock {\em Statistical theory of extreme values and some practical
  applications: a series of lectures}.
\newblock Applied mathematics series. U. S. Govt. Print. Office, 1954.

\bibitem{CVAE}
Eric Jang, Shixiang Gu, and Ben Poole.
\newblock Categorical reparameterization with gumbel-softmax.
\newblock 2017.

\bibitem{trpcage5}
J~Juraszek and PG~Bolhuis.
\newblock Sampling the multiple folding mechanisms of trp-cage in explicit
  solvent.
\newblock {\em Proceedings of the National Academy of Sciences},
  103(43):15859--15864, 2006.

\bibitem{FoldingAtHome}
S.~M. {Larson}, C.~D. {Snow}, M.~{Shirts}, and V.~S. {Pande}.
\newblock {Folding@Home and Genome@Home: Using distributed computing to tackle
  previously intractable problems in computational biology}.
\newblock {\em ArXiv e-prints}, January 2009.

\bibitem{fastfolding}
K.~Lindorff-Larsen, S.~Piana, R.O. Dror, and D.E. Shaw.
\newblock How fast-folding proteins fold.
\newblock {\em Science}, Oct 28;334(6055):517-20., 2011.

\bibitem{Concrete}
Chris~J. Maddison, Andriy Mnih, and Yee~Whye Teh.
\newblock The concrete distribution: {A} continuous relaxation of discrete
  random variables.
\newblock {\em CoRR}, abs/1611.00712, 2016.

\bibitem{ASampling}
Chris~J. Maddison, Daniel Tarlow, and Tom Minka.
\newblock {A* Sampling}.
\newblock In {\em Advances in Neural Information Processing Systems 27}, 2014.

\bibitem{AAE}
Alireza Makhzani, Jonathon Shlens, Navdeep Jaitly, and Ian Goodfellow.
\newblock Adversarial autoencoders.
\newblock In {\em International Conference on Learning Representations}, 2016.

\bibitem{trpcage2}
Fabrizio Marinelli, Fabio Pietrucci, Alessandro Laio, and Stefano Piana.
\newblock A kinetic model of trp-cage folding from multiple biased molecular
  dynamics simulations.
\newblock {\em PLoS computational biology}, 5(8):e1000452, 2009.

\bibitem{MDTraj}
Robert~T. McGibbon, Kyle~A. Beauchamp, Matthew~P. Harrigan, Christoph Klein,
  Jason~M. Swails, Carlos~X. Hern{\'a}ndez, Christian~R. Schwantes, Lee-Ping
  Wang, Thomas~J. Lane, and Vijay~S. Pande.
\newblock Mdtraj: A modern open library for the analysis of molecular dynamics
  trajectories.
\newblock {\em Biophysical Journal}, 109(8):1528 -- 1532, 2015.

\bibitem{MDAnalysis}
Naveen Michaud{-}Agrawal, Elizabeth~J. Denning, Thomas~B. Woolf, and Oliver
  Beckstein.
\newblock Mdanalysis: {A} toolkit for the analysis of molecular dynamics
  simulations.
\newblock {\em Journal of Computational Chemistry}, 32(10):2319--2327, 2011.

\bibitem{NAMD}
J.~C. Phillips, R.~Braun, W.~Wang, J.~Gumbart, E.~Tajkhorshid, E.~Villa,
  C.~Chipot, R.~D. Skeel, L.~Kale, and K.~Schulten.
\newblock {{S}calable molecular dynamics with {N}{A}{M}{D}}.
\newblock {\em J Comput Chem}, 26(16):1781--1802, Dec 2005.

\bibitem{trpcage3}
Jed~W Pitera and William Swope.
\newblock Understanding folding and design: Replica-exchange simulations
  of``trp-cage''miniproteins.
\newblock {\em Proceedings of the National Academy of Sciences},
  100(13):7587--7592, 2003.

\bibitem{amber}
R.~Salomon-Ferrer, D.A. Case, and R.C. Walker.
\newblock An overview of the amber biomolecular simulation package.
\newblock {\em WIREs Comput. Mol. Sci.}, 3:198--210, 2013.

\bibitem{ANTON}
D.~E. Shaw, J.~P. Grossman, J.~A. Bank, B.~Batson, J.~A. Butts, J.~C. Chao,
  M.~M. Deneroff, R.~O. Dror, A.~Even, C.~H. Fenton, A.~Forte, J.~Gagliardo,
  G.~Gill, B.~Greskamp, C.~R. Ho, D.~J. Ierardi, L.~Iserovich, J.~S. Kuskin,
  R.~H. Larson, T.~Layman, L.~S. Lee, A.~K. Lerer, C.~Li, D.~Killebrew, K.~M.
  Mackenzie, S.~Y.~H. Mok, M.~A. Moraes, R.~Mueller, L.~J. Nociolo, J.~L.
  Peticolas, T.~Quan, D.~Ramot, J.~K. Salmon, D.~P. Scarpazza, U.~B. Schafer,
  N.~Siddique, C.~W. Snyder, J.~Spengler, P.~T.~P. Tang, M.~Theobald, H.~Toma,
  B.~Towles, B.~Vitale, S.~C. Wang, and C.~Young.
\newblock Anton 2: Raising the bar for performance and programmability in a
  special-purpose molecular dynamics supercomputer.
\newblock pages 41--53, Nov 2014.

\bibitem{trpcage1}
Christopher~D Snow, Bojan Zagrovic, and Vijay~S Pande.
\newblock The trp cage: folding kinetics and unfolded state topology via
  molecular dynamics simulations.
\newblock {\em Journal of the American Chemical Society}, 124(49):14548--14549,
  2002.

\bibitem{trpcage4}
Ruhong Zhou.
\newblock Trp-cage: folding free energy landscape in explicit water.
\newblock {\em Proceedings of the National Academy of Sciences},
  100(23):13280--13285, 2003.

\bibitem{Survey}
Arthur Zimek, Erich Schubert, and Hans{-}Peter Kriegel.
\newblock A survey on unsupervised outlier detection in high-dimensional
  numerical data.
\newblock {\em Statistical Analysis and Data Mining}, 5(5):363--387, 2012.

\end{thebibliography}
\end{document}